\documentclass[11pt]{article}

\usepackage{amssymb,amsmath,amstext,amsfonts}
\def\cal#1{\fam2#1}
\newcommand\blfootnote[1]{%
  \begingroup
  \renewcommand\thefootnote{}\footnote{#1}%
  \addtocounter{footnote}{-1}%
  \endgroup
}
\newtheorem{theorem}{\sc Theorem}
\newtheorem{lemma}{\sc Lemma}
\newtheorem{prope}{\sc Property}
\newtheorem{coro}{\sc Corollary}
\newtheorem{nota}{\sc Notation}
\newtheorem{defin}{\sc Definition}
\newtheorem{cla}{\sc Claim}
\newtheorem{rem}{\sc Remark}
\newenvironment{proof}{\par \sc Proof.\rm}{\hspace*{\fill}$\Box$\vspace{1ex}}
\newenvironment{remark}{\begin{rem}}{\hspace*{\fill}$\Diamond$\end{rem}}

\newenvironment{corollary}{\begin{coro}}{\end{coro}}

\newenvironment{definition}{\begin{defin}}{\end{defin}}

\def\ds{\displaystyle}

\def\nn{\mathbb N}

\def\nn{\mathbb N}

\newcommand{\ld}[1]{\mbox{depth}_{#1}}
\renewcommand{\emptyset}{\varnothing}

\begin{document}

\title{On Logical Depth and the Running Time of Shortest Programs}
\author{L. Antunes, A. Souto, and P.M.B. Vit\'anyi}
\date{}
\maketitle

\begin{abstract}
The logical
depth with significance $b$ of a finite binary string $x$
is the shortest running time of
a binary program for $x$ that can be compressed by at most 
$b$ bits. There is another definition
of logical depth. We give two theorems about the quantitative relation 
between these versions: the first theorem concerns a variation
of a known fact with a new proof, the second theorem and its proof are new. 
We select the above version of logical depth and show the following.
There is an infinite sequence of strings of increasing length 
such that for each $j$ there is a $b$ such that the logical depth 
of the $j$th string as a function of $j$ is incomputable
(it rises faster than any computable function) 
but with $b$ replaced by $b+1$ the resuling function is computable.
Hence the maximal gap between the logical depths
resulting from incrementing appropriate $b$'s by 1 rises faster than any
computable function.
All functions mentioned are upper bounded by the Busy Beaver function.
Since for every string its logical depth is nonincreasing in $b$,
the minimal computation time of the shortest programs for the sequence 
of strings as a function of $j$ rises faster than any 
computable function but not so fast as the Busy Beaver function.
\end{abstract}

\blfootnote{L. Antunes is with Instituto de Telecomunica\c{c}\~{o}es and
Faculdade de Ci\^{e}ncias Universidade do Porto. He was supported by FCT projects PEst-OE/EEI/LA0008/2011 and PTDC/EIA-CCO/099951/2008. Address: Departamento de Ci\^encia de Computadores R.Campo Alegre, 1021/1055, 4169 - 007 Porto - Portugal. Email: {\tt lfa@dcc.fc.up.pt}
A. Souto is with Instituto de Telecomunica\c{c}\~{o}es and 
Instituto Superior T\'ecnico, Universidade T\'ecnica de Lisboa. He was supported by FCT projects PEst-OE/EEI/LA0008/2011, PTDC/EIA-CCO/099951/2008 and the grant SFRH/BPD/76231/2011. Address:  Departamento de Matem\'atica IST Av. Rovisco Pais, 1, 1049-001 Lisboa - Portugal. Email: {\tt a.souto@math.ist.utl.pt}\\
P.M.B. Vit\'anyi is with CWI and University of Amsterdam. Address:
CWI, Science Park 123, 1098XG Amsterdam, The Netherlands. Email: {\tt Paul.Vitanyi@cwi.nl}
}
\section{Introduction}

The logical depth is related to complexity with bounded resources.
Computing a string $x$ from one of its shortest programs 
may take a very long time.
However, computing the same string from a simple `print$(x)$' program of length about $|x|$ bits takes very little time.

A program for $x$ of larger length than a given program 
for $x$ may decrease the computation time but in general
does not increase it.
Exceptions are, for example,
cases where unnecessary steps are considered. 
Generally we associate
the longest computation time with a shortest
program for $x$. There arises the question how much time can be saved 
by computing a given string from a longer program.

\subsection{Related Work}
The minimum time to compute a string by a $b$-incompressible 
program was first considered in
\cite{ben88} Definition 1. The minimum time
was called the {\em logical depth at significance 
$b$} of the string concerned.
Definitions, variations, discussion and early results can be found
in the given reference.
A more formal treatment, as well as an intuitive approach, was given in
the textbook \cite{LiVi}, Section 7.7. 
In \cite{afmv06} the notion of {\em computational}
depth is defined as $K^d(x) -K(x)$.
This would or would not equal the negative logarithm of the expression
$Q^d(x)/Q(x)$ in Definition~\ref{def.ld} as follows.
In \cite{lev74} L.A. Levin proved, in the so-called Coding Theorem 
\begin{equation}\label{eq.coding}
-\log Q(x) = K(x)+O(1)
\end{equation}
(see also \cite{LiVi} Theorem 4.3.3).
It remains to prove or disprove
$-\log Q^d (x) = K^d(x)$ up to a small additive term:
a major open problem in Kolmogorov complexity theory,
see \cite{LiVi} Exercises 7.6.3 and 7.6.4.
For Kolmogorov complexity notions see Section~\ref{sect.kolm},
and for $Q$ and $Q^d$ see \eqref{eq.U}.

\subsection{Results}
There are two versions of logical depth, Definition~\ref{def.ld}
and Definition~\ref{def.final}. 
The two versions are related.
The version of Definition~\ref{def.final} almost implies that of
Definition~\ref{def.ld} (Theorem~\ref{Teo1}), but
vice versa there is possible uncertainty (Theorem~\ref{Teo2}).
We use Definition~\ref{def.final}, that is, $\ld{b}^{(2)}(x)$.
There is an infinite sequence of strings 
$x_1,x_2, \ldots$ with $|x_{j+1}|=|x_j|+1$ and an infinite
sequence of positive integers $b_1,b_2, \ldots$, which satisfy the following.
For every $j>0$ the string $x_j$ is computed by two
programs that can be compressed by at most
$b_j,b_j+1$-bits and take least computation time
among programs of their lengths,
respectively. 
Let these computation times be
$d_1^j,d_2^j$ steps. Then the function
$h(j)=d_1^j-d_2^j$ rises faster than any computable function 
but not as fast as the Busy Beaver function,
the first incomputable function \cite{Ra62} 
(Theorem~\ref{theo.3} and Corollary~\ref{cor.2}). 
For the associated shortest programs $x^*_1,x^*_2, \ldots$ of $x_1,x_2, \ldots$ 
the function $s^*(j)$ defined as the minimum number of steps in the
computation of $x^*_j$ to $x_j$ ($j > 0$). Then the function $s^*$
rises faster than any computable function but again not
so fast as the Busy Beaver function (Corollary~\ref{theo.BB}) . 

The rest of the paper is organized as follows. Section~\ref{sect.2} 
introduces notation, definitions and basic results needed 
for the paper. Section~\ref{sect.3} defines two versions of logical depth 
and proves quantitative relations between them.
In Section~\ref{instabilidade}, we prove the other results mentioned.

\section{Preliminaries}\label{sect.2}
We use {\em string} or {\em program} to mean a finite binary string.
Strings are denoted by the 
letters $x$, $y$ and $z$. The {\em length} of a 
string $x$ (the number of occurrences of bits in it) is
denoted by $|x|$, and the {\em empty} string by $\epsilon$. Thus,
$|\epsilon|=0$.
The notation ``$\log$'' means the binary logarithm.
Given two functions $f$ and $g$, we say that $f \in O(g)$ if there is a constant $c >0$, 
such that $f(n) \leq c \cdot g(n)$, 
for all but finitely many natural numbers $n$. 
Restricting the computation time resource is
indicated by a superscript giving the allowed number of steps, usually
using $d$.

\subsection{Computability}
A pair of nonnegative integers,
such as $(p,q)$ can be interpreted as the rational $p/q$.
We assume the notion of a computable function with rational arguments
and values.
A function $f(x)$ with $x$ rational is \emph{semicomputable from below}
if it is defined by a rational-valued total computable function $\phi(x,k)$
 with $x$ a rational number
and $k$ a nonnegative integer
such that $\phi(x,k+1) \geq \phi(x,k)$ for every $k$ and
  $\lim_{k \rightarrow \infty} \phi (x,k)=f(x)$.
This means
that $f$ (with possibly real values)
can be computed in the limit from below
(see \cite{LiVi}, p. 35).  A function $f$ is  \emph{semicomputable
from above} if $-f$ is semicomputable from below.
 If a function is both semicomputable from below
and semicomputable from above then it is \emph{computable}.

\subsection{Kolmogorov Complexity}\label{sect.kolm}
We refer the reader to the textbook
\cite{LiVi} for details, notions, and history.
We use Turing machines with a read-only one-way input tape, one or more
(a finite number) of two-way work tapes at which the computation takes place,
and a one-way write-only output tape. All tapes are  semi-infinite
divided into squares, and each square can contain a symbol. 
Initially, the input tape is inscribed with a 
semi-infinite sequence of 0's and 1's.
The other tapes are empty (contain only blanks).
At the start, all tape heads scan 
the leftmost squares on their tapes. If the machine
halts for a certain input then the contents of the scanned segment
of input tape is called the {\em program} or {\em input}, and the contents of
the output tape is called the {\em output}. The machine thus described
is a {\em prefix Turing machine}. Denote it by $T$.
If $T$ terminates with program $p$ then the output is $T(p)$. The set 
${\cal P}=\{p:T(p)<\infty \}$ is {\em prefix-free} (no element of the set
is a proper prefix of another element). By the ubiquitous Kraft inequality
\cite{Kr49} we have 
\begin{equation}\label{eq.kraft}
\sum_{p \in {\cal P}} 2^{-|p|} \leq 1.   
\end{equation}
We extend the prefix Turing machine with an extra read-only tape called 
the {\em auxiliary} or {\em conditional}. Initially it contains the auxiliary 
information consisting of a string $y$. We write $T(p,y)$ and the
set ${\cal P}_y=\{p:T(p,y)<\infty \}$ is also prefix-free.
The relation \eqref{eq.kraft} holds also with ${\cal P}_y$ substituted
for ${\cal P}$ and $y$ 
is fixed auxiliary information. 
The unconditional case 
corresponds to the case where the conditional is $\epsilon$.

If $T_1,T_2, \ldots$ is a standard enumeration of prefix Turing machines,
then certain of those are called universal. {\em Universal} prefix Turing
machines are those that can simulate any other machine in the enumeration.
Among the universal prefix Turing machines we consider a special
subclass called {\em optimal}, see Definition 2.0.1 in \cite{LiVi}.
To illustrate this concept let $T_1,T_2, \ldots$ be a standard 
enumeration of prefix
Turing machines, and let $U_1$ be one of them.  If $U_1 (i,pp) = T_i(p)$
for every index $i$ and program $p$ and outputs 0 for inputs that are
not of the form $pp$ (doubling of $p$), 
then $U_1$ is also universal. However,
$U_1$ can not be used to define Kolmogorov complexity. For that we need 
a machine $U_2$ with $U_2(i,p)=T_i(p)$ for every $i,p$. 
The machine $U_2$ is called an {\em optimal} prefix Turing machine.
Optimal prefix Turing machines are a strict subclass of 
universal prefix Turing machines. The above
example illustrates the strictness. 
The term `optimal' comes from the founding paper \cite{Ko65}.

It is possible that two different optimal prefix Turing machines 
have different computation times for the same input-output pairs or
they have different sets of programs. To avoid these problems we
fix a reference machine. Necessarily, the reference machine has a 
certain number of worktapes.
A well-known result of \cite{HeSt66} states that $n$ steps of 
a $k$-worktape
prefix Turing machine can be simulated in $O(n \log n)$ steps
of a two-worktape prefix Turing machine 
(the constant hidden in the big-$O$ notation
depends only on $k$). 
Thus,
for  such a simulating optimal Turing machine $U$ we have
$U(i,p)=T_i(p)$ for all $i,p$; if $T_i(p)$ terminates in time $t(n)$
then $U(i,p)$ terminates in time $O(t(n)\log t(n))$. 
Altogether, we fix such a simulating 
optimal prefix Turing machine and call it the {\em reference
optimal prefix Turing machine} $U$.
\begin{definition}
\rm
Let $U$ be the reference optimal prefix Turing machine, 
and $x,y$ be strings. 
The \em prefix Kolmogorov complexity \em 
$K(x|y)$ of $x$ given $y$ is defined by 
$$
  K(x|y)=\min \{|q| : U(q,y)= x \}.
$$
(Earlier we wrote $U(i,p)$ while we write here $U(q,y)$. The two
are reconciled by writing $i,p=i,r,y=q,y$. That is, $p=r,y$ for a program
$r$, and $q=i,r$.) 

The notation $U^d(q,y)=x$ means that $U(q,y)=x$ within
$d$ steps.
The \em $d$-time-bounded
prefix Kolmogorov complexity \em $K^{d}(x|y)$ of $x$ given $y$ is
defined by
$$K^{d}(x|y) =\min \{|q| : U^d(q,y)=x \}.$$
\end{definition}
The default value for the auxiliary input $y$ for the program $q$, 
is the empty string $\epsilon$. To avoid overloaded notation we usually 
drop this argument in case it is there. 
Let $x$ be a string. Denote by $x^*$ the first shortest program in standard
enumeration such that $U(x^*)=x$.
A string is {\em $c$-incompressible} if a shortest program for it
is at most $c$ bits shorter than the string itself.

\section{Different Versions of Logical Depth}\label{sect.3}
The logical depth \cite{ben88} comes in two versions.
One version is based on $Q_U(x)$, 
the so-called {\em a priori} probability \cite{LiVi} and its 
time-bounded version $Q_U^d$. Here $U^d(p)$ means that $U(p)$ terminates in
at most $d$ steps. 
For convenience we drop the subscript on $Q_U$ and $Q^d_U$ 
and consider $U$ as understood.
\begin{equation}\label{eq.U}
Q(x)=\sum_{U(p)=x}2^{-|p|}, \;\;\;\;\;Q^d(x)=\sum_{U^d(p)=x}2^{-|p|}.
\end{equation}
\begin{definition}\label{def.ld}
\rm
Let $x$ be a string, $b$ a nonnegative integer.
The logical depth, version 1, 
of $x$ at significance level $\varepsilon= 2^{-b}$ is
\[
\ld{\varepsilon}^{(1)}(x) = \min \left\{d:\ds\frac{Q^d(x)}{Q(x)}\geq \varepsilon\right\}.
\]
\end{definition}
Using a program that is longer than another program for output $x$
can shorten the computation time. 
The $b$-significant logical depth of an object $x$ can also be defined 
as the minimal time 
the reference optimal prefix Turing machine needs 
to compute $x$ from a program which is $b$-incompressible.
\begin{definition}\label{def.final}
\rm
Let $x$ be a string, $b$ a nonnegative integer. 
The logical depth, version 2, of $x$ at significance level 
$b$, is: 
\[
\ld{b}^{(2)}(x) = \min \{d: |p| \leq K(p) + b \wedge U^d(p) = x\} \,.
\] 
\end{definition}
\begin{remark}
\rm
The program $x^*$ is the first shortest
program for $x$ in enumeration order. It may not be the fastest shortest program
for $x$. Therefore, if $U^d(x^*)=x$ then $d \geq \ld{0}^{(2)}(x)$. 
For $b > 0$ the value of $\ld{b}^{(2)}(x)$
is monotonic nonincreasing until 
$$\ld{|x|-K(x)+O(1)}^{(2)}(x)=O(|x| \log |x|),$$
where the $O(1)$ term represents the length of
a program to copy the literal representation
of $x$ in $O(|x| \log |x|)$ steps. 
If $x$ is random
($|x|=n$ and $K(x) \geq n$) then for $b=O(\log n)$ 
we have $\ld{b}^{(2)}(x)=O(n \log n)$---we print a literal copy of $x$.
These $x$'s, but not only these, are called {\em shallow}.
\end{remark}

For version (2)
every program $p$ of length at most $K(p)+b$ 
must take at least $d$ steps
to compute $x$. 
Version (1) states that 
$Q^d(x)/Q(x) \geq 2^{-b}$ and 
$Q^{d-1}(x)/Q(x) < 2^{-b}$.
Statements similar 
to Theorem~\ref{Teo1} and Remark~\ref{rem.Teo1} were shown in
Theorem 7.7.1 and Exercise 7.7.1 in \cite{LiVi} 
and derive from \cite{ben88} Lemma 3.
\begin{theorem}\label{Teo1}
If $\ld{b}^{(2)}(x)=d$ then
$\ld{2^{-\beta}}^{(1)}(x)=d$ with $b+1 < \beta \leq b+K(b)+O(1)$.
\end{theorem}
\begin{proof}
The theorem states: 
if $\ld{b}^{(2)}(x)=d$ then
\[
\frac{1}{2^{b+K(b)+O(1)}} \leq \frac{Q^d(x)}{Q(x)}
< \frac{1}{2^{b+1}}.
\]
(Right $<$) 
By way of contradiction $Q^d(x) \geq 2^{-b-1}Q(x)$.
If for a nonnegative constant $c$
all programs computing $x$ within $d$ steps are $c$-compressible, 
then the twice iterated
reference optimal Turing machine (in its role as decompressor) computes
$x$ with probability $2^c Q^d(x) \geq 2^{c-b-1}Q(x)$ 
from the $c$-compressed versions.
But $Q(x) \geq 2^{c-b-1}Q(x)+2^{-b-1} Q(x)$.
Therefore $c-b-1 < 0$  that is $c < b+1$. This implies that
there is a program computing
$x$ within $d$ steps that is $(b+1)$-incompressible. Then
$\ld{b+1}^{(2)}(x)=d$ contradicting the assumption that
$\ld{b}^{(2)}(x)=d$.
Hence $Q^d(x) < 2^{-b-1} Q(x)$.
(Left $\leq$) 
By way of contradiction 
$Q^d(x) < 2^{-B} Q(x)$ with $B=b+K(b)+c$ and $c$ is
a large enough constant to derive the contradiction below.
Consider the following lower semicomputable
semiprobability (the total probability is less than 1):
For every string $x$ we enumerate
all programs $p$ that compute $x$ in order of halting (time), and
assign to each halting $p$ the probability $2^{-|p|+B}$ until
the total probability would pass $Q(x)$ with the next halting $p$.
Since $Q(x)$ is lower semicomputable we can postpone 
assigning probabilities.
But eventually or never for a program $p$
the total probability may pass $Q(x)$ and this $p$ and all subsequent
halting $p$'s for $x$ get assigned probability 0.   
Therefore, the total probability assigned to all halting programs
for $x$ is less than $Q(x)$. Since by contradictory
assumption $Q^d(x) < 2^{-B} Q(x)$ we have
$2^B Q^d(x) = \sum_{U^d(p)=x} 2^{-|p|+B} < Q(x)$. Since $Q(x) < 1$ we
have $Q^d(x) < 2^{-B}$ and therefore all programs that
compute $x$ in at most $d$ steps are $(B-O(1))$-compressible
given $B$, and therefore $(B-K(B)-O(1))$-compressible.

Since $\ld{b}^{(2)}(x)=d$, there exists a $b$-incompressible 
program from which $x$ can be computed in $d$ steps. 
Since $K(b+K(b)) \leq K(b,K(b))+O(1) \leq K(b)+O(1)$
by an easy argument \cite{LiVi} and $K(c) = O(\log c) < c/2$
we have that
$B-K(B)-O(1)=
b+K(b)+c-K(b+K(b)+c)-O(1) \geq b+K(b)+c-K(b+K(b))-K(c)-O(1)>b$.
This gives the required contradiction.
Hence $Q^d(x) \geq 2^{-B} Q(x)$.
\end{proof}

\begin{remark}\label{rem.Teo1}
\rm
We can replace $K(b)$ by $K(d)$ by changing the construction of the
semiprobability: knowing $d$ we generate all programs that compute 
$x$ within $d$ steps and let the semiprobabilities be proportional to 
$2^{-|p|}$ and the sum be at most $Q(x)$. In this way $K(b)$ 
in Theorem~\ref{Teo1} is substituted by $\min\{K(b),K(d)\}$. 
\end{remark}
\begin{remark}\label{rem.conv}
\rm
Possibly $Q^d(x)=Q^{d+1}(x)$. 
Moreover, while $Q^d(x)/Q(x) \geq 2^{-b}$ 
for $d$ least, possibly also $Q^d(x)/Q(x) \geq 2^{-b+1}$.
Both events happen, for example, if $p$ computes $x$ in $d$ steps but
not in $d-1$ steps, there is no program for $x$ halting in $d+1$ steps,
and $|p| < b-1$ while $Q^{d-1} < 2^{-b}$. 
In the next theorem if we write $\ld{2^{-b}}^{(2)}(x)=d$ then 
$d,b$ are least integers for which this equality holds. 
\end{remark}
\begin{theorem}\label{Teo2}
Let $U^d(p)=U^d(q)=U^d(r)=x$ and $U^{d-1}(p),U^{d-1}(r) = \infty$ 
(the  computation does not
halt in $d-1$ steps) with $K(p)$ least, $|q| \leq K(x)+b$, 
and $||r|-|q||$ (the absolute value of the difference in lengths) 
is minimal. 
If $\ld{2^{-b}}^{(1)}(x)=d$ then
$\ld{\beta}^{(2)}(x)=d$ with 
$b-K(p)+K(x)-O(1) \leq \beta \leq b+||r|-|q||-O(1)$.
\end{theorem}
\begin{proof}
(Left $\leq$) By assumption 
$Q^d(x)/Q(x) \geq 2^{-b}$ and by 
Remark~\ref{rem.conv} we have $Q^d(x)/Q(x) < 2^{-b+1}$. 
Since it is easy to see that $2^{-K(x)} < Q(x)$, we
have $Q(x)=2^{-K(x)+c}$ for a positive constant $c$ 
by \eqref{eq.coding}. Therefore
\[
Q^d(x)= \sum_{U^d(p)=x} 2^{-|p|} < 2^{-(K(x)+b-1-c)}.
\]
Hence every program $p$ such that $U^d(p)=x$ satisfies
$K(x)+b-1-c < |p|$. Denote the set of these prograns by $D$. 
Let $p \in D$ be such that 
$U^{d-1}(p)= \infty$.
Let $(D-1)$ be the set of such programs. Then $(D-1) \subseteq D$.
By the convention in Remark~\ref{rem.conv} we have
$(D-1) \neq \emptyset$. 
Since $K(x) \leq K(p)+O(1)$ for all $p$ satisfying $U(p)=x$
(the $O(1)$ term is a nonnegative constant
independent of $x$ and $p$) we have
$K(p)+(b-f_x(d)) < |p|$ for all $p \in (D-1)$, with
$f_x(d)=K(p)-K(x)+O(1) >0$.
If $\ld{\beta}^{(2)}(x)=d$ then $|p|\leq K(p)+\beta$ for $p \in (D-1)$.
Hence $\beta \geq b-f_x(d)$.

(Right $\leq$):
By assumption
$Q^d(x) \geq  2^{-b} Q(x)$. 
Let $P$ be the set of programs $p$ such that $U(p)=x$, 
the set $Q$ consist of programs $q \in P$ such that
$|q| \geq |p|+b$ with $|q|$ least and $p \in P$, while the sets $D,(D-1)$
are defined above. 
Then $D,Q \subseteq P$, $(D-1) \subseteq D$, and 
\begin{align*}
Q^d(x) &= \sum_{U^d(p)=x} 2^{-|p|}  
\geq 2^{-b} \sum_{U(p)=x} 2^{-|p|}
\\&=\sum_{p \in P} 2^{-|p|-b}
\geq \sum_{q \in Q} 2^{-|q|}.
\end{align*}
The last sum is at most the first sum and the programs of $Q$
constitute all the programs in $P$ that have length at least
$K(x)+b$ (a shortest program in $P$ trivially having length $K(x)$ and
therefore a shortest program in $Q$ has length $K(x)+b$).
Since $D \subseteq P$ either $D=Q$
or $D \bigcap Q \neq \emptyset$.
It follows that
there exist programs $q \in D$ at least as short as 
the shortest program of $Q$. 
Since a shortest program in $Q$ has length $K(x)+b$ 
therefore $|q| \leq K(x)+b$. 
By the convention of Remark~\ref{rem.conv} $(D-1) \neq \emptyset$. 
Choose $r \in (D-1)$ and $q \in D$ with $|q| \leq K(x)+b$ such that
$||r|-|q||$ is minimal. Additionally,
$K(r) \geq K(x)+O(1)$ (since $U(r)=x$ and an $O(1)$ term
independent of $r$ and $x$). Therefore
$|r| \leq K(r) +b +g_x(d)$ with 
$g_x(d)= ||r|-|q||-O(1)$.
It follows that
if $\ld{\beta}^{(2)}(x)=d$ then $\beta \leq b+g_x(d)$. 
\end{proof}

According to Theorems~\ref{Teo1}, \ref{Teo2} version 2 
implies version 1 with the same depth $d$ and almost the
same parameter $b$, while 
version 1 implies version 2 with the same depth $d$
but more uncertainty in the parameter $b$. 
We choose version 2 as our final definition of logical depth.

\section{The graph of logical depth}\label{instabilidade}

Even slight changes of the 
significance level $b$ can cause large changes
in logical depth.

\begin{lemma}\label{lemma1}
Every function $\phi$ such that
$\phi(x) \geq \min\{d: U^d(p)=x ,\; |p|=K(x)\}$ 
is incomputable and
rises faster than any computable function.
\end{lemma}
\begin{proof}
By \cite{Ga74} we have $K(K(x)|x) \geq \log n - 2 \log \log n-O(1)$.
(This was improved to the optimal $K(K(x)|x) \geq \log n -O(1)$ recently
in \cite{BS03}.) Hence there is no computable function $\phi(x) \geq
\min\{d: U^d(p)=x ,\; |p|=K(x)\}$. If there were, then we could run
$U$ for $d$ steps on any program of length $n+O(\log n)$. Among the 
programs that halt within $d$ steps we select the ones which output $x$.
Subsequently, we select from this set a program of minimum length.
This is a shortest program for $x$ of length $K(x)$. Therefore,
the assumption that $\phi$ is computable implies that $K(K(x)|x) = O(1)$ and hence a contradiction.
\end{proof}

\begin{definition}\label{bbfunction}
\rm
The \em Busy Beaver function \em $BB: \nn \rightarrow \nn$ is defined by
\[
BB(n) = \max \{d:|p|\leq n \wedge U^d (p) < \infty\}
\]
\end{definition}

The following result was mentioned informally in \cite{ben88}.
\begin{lemma}\label{lemma2}
The running time of a program $p$ 
is at most $BB(|p|)$. 
The running time of a shortest program for a string $x$ of length $n$
is at most $BB(n+O(\log n))$.
\end{lemma}
\begin{proof}
The first statement of the lemma follows from 
Definition~\ref{bbfunction}. For the second statement we use
the notion of a simple prefix-code called a {\em self-delimiting} code. 
This is obtained by reserving one
symbol, say 0, as a stop sign and encoding a string $x$ as $1^x 0$.
We can prefix an object with its length and iterate
this idea to obtain ever shorter codes:
$\bar{x}= 1^{|x|} 0 x$ with
length $|\bar{x}| = 2|x| + 1$, and 
$x'= \overline{|x|} x$ of length $|x|+2||x||+1=
|x|+O(\log |x|)$ bits. From this code $x$ is readily extracted.
The second statement follows since $K(x) \leq |x|+O(\log |x|)$. 
\end{proof}

\begin{theorem}\label{theo.3}
There is an infinite sequence of strings $x_1,x_2, \ldots$ with
$|x_{j+1}|=|x_j|+1$ ($j \geq 1$) and an infinite sequence 
$b_1,b_2, \ldots$ 
of integers such that $f(j)=\ld{b_j}^{(2)}(x_j)$ is incomputable (faster
than any computable function)  and
$g(j)=\ld{b_j+1}^{(2)}(x_j)$ is computable. 
\end{theorem}

\begin{proof}
Let $\phi(x)\geq \ld{0}^{(2)}(x)$ be an incomputable function 
as in Lemma~\ref{lemma1}. 
The function $\psi$ defined by
$\psi(x)=\ld{n-K(x)+O(\log n)}^{(2)}(x)= O(n \log n)$ for $|x|=n$ is 
computable. Namely,
a self-delimiting encoding
of $x$ can be done in $n+O(\log n)$ bits. Let $q$ be such an encoding
with $q=1^{||x||}0|x|x$ (where $||x||$ is the length of $|x|$).
Let $r$ be a self-delimiting program of $O(1)$ bits which prints the
encoded string. Consider the program $rq$.
Since $x$ can be compressed to length $K(x)$, 
the running time  $\ld{n-K(x)+O(\log n)}^{(2)}(x)$ is at most
the running time of $rq$ which is $O(n \log n)$.
\end{proof}
\begin{corollary}\label{cor.2}
\rm
Define the function $h$ by
$h(j)= f(j)-g(j)$.
Then $h$ is a gap in the logical depths of which 
the significance differs by 1.
The function $h(j)$ rises faster than any computable function but not
faster than $BB(|x_j|+O(\log |x_j|)$ by Lemma~\ref{lemma2}.
\end{corollary}
\begin{corollary}\label{theo.BB}
\rm
Let $s^*(j)=\ld{O(1)}^{(2)}(x_j)$ be the minimal time of a computation 
of a shortest program for $x_j$, and
$f$ be the function in statement of Theorem~\ref{theo.3}. Then
$f(j) \leq s^*(j)  \leq BB(|x_j|+O(\log |x_j|))$.
\end{corollary}
Namely, the logical depth function $\ld{b}^{(2)}(x)$ 
is monotonic nonincreasing in the significance argument $b$
for all strings $x$ by its Definition~\ref{def.final}. 
By Lemma~\ref{lemma2} and Corollary~\ref{cor.2} the 
Corollary~\ref{theo.BB} follows.

\section{Conclusion}
We resolve quantitative relations between
the two versions of logical depth in the literature. 
One of these relations was known by another proof, the
other relation is new.
We select one version that approximately implies the other,
and study the
the behavior of the resulting 
logical depth function associated
with a string $x$ of length $n$. 
This function is monotonic nonincreasing. For
argument 0 the logical depth is at least the minimum running time of the
computation from a shortest program for $x$. The
function decreases to $O(n \log n)$ for the argument 
$|x|-K(x)+O(\log |x|)$.
We show that there is an infinite sequence of strings such that 
maximum gap of
logical depths resulting from consecutive significance levels
rises faster than any computable function,
that is, incomputably fast, but not more than the Busy Beaver function.
This shows that 
logical depth can increase tremendously for
only an incremental difference in significance.
Moreover, the
minimal computation times of associated shortest 
programs rises incomputably fast  but not so fast as the 
Busy Beaver function.

\section*{Acknowledgment}
We thank Anonymus for comments and the new proof of Theorem~\ref{Teo1}.

\bibliographystyle{plain}

\end{document}